\documentclass[a4paper,twocolumn,10pt,accepted=2020-11-26]{quantumarticle}

\pdfoutput=1
\usepackage{bbm}

\usepackage[numbers,sort&compress]{natbib}

\usepackage[utf8]{inputenc}
\usepackage[english]{babel}
\usepackage[T1]{fontenc}

\newcommand{\reed}{{\mathrm{RM}}}
\newcommand{\Res}{{\mathrm{Res}}}
\newcommand{\GA}{{\mathrm{GA}}}
\newcommand{\sperp}{{{\perp_{\rm{s}}}}}

\usepackage{graphicx,epstopdf,epsfig}
\usepackage{tikz-cd} 
\usepackage{inputenc}
\usepackage{hyperref, amsmath,  amssymb,amsbsy,amsfonts,amsthm,latexsym,amsopn,amstext, amsxtra,euscript,amscd}

\usepackage{hypernat}
\usepackage{accents}
\hypersetup{
  colorlinks   = true, 
  urlcolor     = blue, 
  linkcolor    = blue, 
  citecolor   = blue 
}

\newtheorem{theo}{Theorem}
\newtheorem{prop}[theo]{Proposition}
\newtheorem{cor}[theo]{Corollary}
\newtheorem{lem}[theo]{Lemma}
\theoremstyle{definition}

\newtheorem{exa}[theo]{Example}

\newtheorem{rem}[theo]{Remark}

\def\l{\langle}
\def\r{\rangle}

\newcommand{\U}{{\mathbb U}}
\newcommand{\F}{{\mathbb F}}

\newcommand{\Z}{{\mathbb Z}}

\newcommand{\C}{{\mathbb C}}

\newcommand{\cC}{{\mathcal C}}

\newcommand{\cE}{{\mathcal E}}
\newcommand{\cF}{{\mathcal F}}

\newcommand{\cH}{{\mathcal H}}

\newcommand{\cM}{{\mathcal M}}

\newcommand{\cP}{{\mathcal P}}

\newcommand{\cW}{{\mathcal W}}

\newcommand{\bA}{{\mathbf A}}
\newcommand{\bB}{{\mathbf B}}
\newcommand{\bC}{{\mathbf C}}
\newcommand{\bD}{{\mathbf D}}
\newcommand{\bE}{{\mathbf E}}
\newcommand{\bF}{{\mathbf F}}
\newcommand{\bG}{{\mathbf G}}
\newcommand{\bH}{{\mathbf H}}
\newcommand{\bI}{{\mathbf I}}

\newcommand{\bM}{{\mathbf M}}
\newcommand{\bN}{{\mathbf N}}

\newcommand{\bP}{{\mathbf P}}

\newcommand{\bS}{{\mathbf S}}
\newcommand{\bT}{{\mathbf T}}
\newcommand{\bU}{{\mathbf U}}

\newcommand{\bW}{{\mathbf W}}

\newcommand{\ba}{{\mathbf a}}
\newcommand{\bb}{{\mathbf b}}
\newcommand{\bc}{{\mathbf c}}
\newcommand{\bd}{{\mathbf d}}
\newcommand{\be}{{\mathbf e}}

\newcommand{\bv}{{\mathbf v}}
\newcommand{\bw}{{\mathbf w}}
\newcommand{\bx}{{\mathbf x}}

\newcommand{\HW}{\cH\cW_{\!N}}
\newcommand{\Fm}{\F_2^{2m}}

\newcommand{\sbt}{\raisebox{.2ex}{\mbox{$\scriptscriptstyle\bullet\,$}}}

\newcommand{\Sp}{\mbox{\rm Sp}}

\newcommand{\supp}{\mbox{\rm supp}}

\newcommand{\rk}{\mbox{{\rm rank\,}}}

\newcommand{\cl}{\mbox{${\rm Cliff\!}$}}
\newcommand{\im}{\mbox{\rm im}\,}

\newcommand{\rs}{\mbox{\rm rs}\,}

\newcommand{\Tr}{\mbox{\rm Tr}}

\newcommand{\hwt}{\mbox{${\rm wt}_{\rm H}$}}

\newcommand{\T}{\mbox{\rm $^{\scriptsize{\text{t}}}$}}

\newcommand{\inner}[2]{\mbox{$\langle\,{#1}\,|\,{#2}\,\rangle$}}
\newcommand{\inners}[2]{\mbox{$\langle{\,{#1}\,}|\,{#2}\,\rangle_{\rm s}$}}

\newcommand{\GL}{\mathrm{GL}}
\newcommand{\Sym}{\mathrm{Sym}}

\newcommand{\Fix}{\mathrm{Fix}}

\newcommand{\diag}{\textup{diag}}

\newcommand{\Imr}{\bI_{m|r}}
\newcommand{\Imrr}{\bI_{m|-r}}

\newcommand{\fourmat}[4]{\mbox{$\left[\!\!\begin{array}{cc}{#1}&{#2}\\{#3}&{#4}\end{array}\!\!\right]$}}

\newcounter{alp}
\newcounter{ara}
\newcounter{rom}

\newenvironment{arabiclist}{\begin{list}{(\arabic{ara})\hfill}{\usecounter{ara}
     \topsep0.4ex \labelwidth.6cm \leftmargin.6cm \labelsep0cm
     \rightmargin0cm \parsep0ex \itemsep0ex}}{\end{list}}

\begin{document}
\title{Un-Weyl-ing the Clifford Hierarchy}

\author{Tefjol Pllaha}
\affiliation{Department of Communications and Networking, Aalto University, Finland}
\email{tefjol.pllaha@aalto.fi}
\author{Narayanan Rengaswamy}
\affiliation{Department of Electrical and Computer Engineering, University of Arizona, Tucson, AZ, USA}
\email{narayananr@arizona.edu}
 \thanks{Most of this work was conducted when N.R. was with the Department of Electrical and Computer Engineering, Duke University, NC, USA}
\author{Olav Tirkkonen}
\affiliation{Department of Communications and Networking, Aalto University, Finland}
\email{olav.tirkkonen@aalto.fi}
\author{Robert Calderbank}
\affiliation{Department of Electrical and Computer Engineering, Duke University, NC, USA}
\email{robert.calderbank@duke.edu}

\maketitle

\begin{abstract}
The teleportation model of quantum computation introduced by Gottesman and Chuang (1999) motivated the development of the Clifford hierarchy.
Despite its intrinsic value for quantum computing, the widespread use of magic state distillation, which is closely related to this model, emphasizes the importance of comprehending the hierarchy.
There is currently a limited understanding of the structure of this hierarchy, apart from the case of diagonal unitaries (Cui et al., 2017; Rengaswamy et al. 2019).
We explore the structure of the second and third levels of the hierarchy, the first level being the ubiquitous Pauli group, via the Weyl (i.e., Pauli) expansion of unitaries at these levels.
In particular, we characterize the support of the standard Clifford operations on the Pauli group.
Since conjugation of a Pauli by a third level unitary produces traceless Hermitian Cliffords, we characterize their Pauli support as well.
Semi-Clifford unitaries are known to have ancilla savings in the teleportation model, and we explore their Pauli support via symplectic transvections.
Finally, we show that, up to multiplication by a Clifford, every third level unitary commutes with at least one Pauli matrix. 
This can be used inductively to show that, up to a multiplication by a Clifford, every third level unitary is supported on a maximal commutative subgroup of the Pauli group.
Additionally, it can be easily seen that the latter implies the generalized semi-Clifford conjecture, proven by Beigi and Shor (2010).
We discuss potential applications in quantum error correction and the design of flag gadgets.
\end{abstract}



\section{Introduction}

Quantum computing provides a fundamentally new approach to computation by exploiting the laws of quantum mechanics that govern our universe.
In this computational model, a quantum circuit consists of a sequence of operations each of which is either a \emph{quantum gate}, characterized by a unitary matrix, or a \emph{quantum measurement}, characterized by a Hermitian matrix (i.e., an \emph{observable})~\cite{NC00}.
So, a \emph{universal} quantum computer must be capable of implementing \emph{arbitrary} unitary operations and measuring any Hermitian operator on a given set of $m$ qubits. 
In 1999, Gottesman and Chuang demonstrated that such universal quantum computing can be performed just by using the quantum teleportation protocol if one has access to certain standard resources --- Bell-state preparation, Bell-basis measurements, and arbitrary single-qubit rotations~\cite{Gottesman-nature99}.
They defined the \emph{Clifford hierarchy} as part of their proof, and this has proven to be a useful characterization of a large set of unitary operations, both in theory and practice.
In fact, in their teleportation model of computation, the level of a unitary in the hierarchy can be interpreted as a measure of complexity of implementing it.
Furthermore, this model is closely related to the currently widespread scheme of distilling ``magic'' states and injecting them via teleportation-like methods in order to fault-tolerantly execute unitary operations on qubits encoded in a quantum error-correcting code~\cite{Bravyi-pra05,Bravyi-pra12}.
Hence, it is very important to understand the structure of this hierarchy since it has important implications for fault-tolerant quantum computing.

The first level of the hierarchy is the \emph{Pauli} (or \emph{Heisenberg-Weyl}) group and the second level is the \emph{Clifford} group, which is defined as the \emph{normalizer} of the Pauli group in the unitary group.
Subsequent levels $\cC^{(k)}$ of the hierarchy, for $k \geq 2$, are defined recursively as those unitaries that map Pauli matrices to $\cC^{(k-1)}$ under conjugation~\cite{Gottesman-nature99} (see~\eqref{eq:cliff_hierarchy} for the precise definition).
While the first two levels form groups, it is known that the higher levels are only finite sets of unitary matrices (up to overall phases) and that even when $k \rightarrow \infty$, the hierarchy does not encompass all unitary matrices (see Example~\ref{ex:UoutsideCliffHrchy}).
Furthermore, each level is closed under left or right multiplication by Cliffords~\cite{ZCC08}.

It is well-known that the Pauli matrices form an orthonormal basis for all square matrices under the Hilbert-Schmidt inner product~\cite{Gottesman-phd97}.
Therefore, a natural question to consider is to determine the Pauli (i.e., Weyl) expansion of all unitaries in the Clifford hierarchy.
It is reasonable to expect that the Pauli expansion of elements at a level provides insight into the structure of the hierarchy.
Indeed, a subset of the current authors recently identified a special set of diagonal unitary matrices in the hierarchy, called \emph{Quadratic Form Diagonal (QFD)} gates, and produced formulae for their action on Pauli matrices~\cite{RCP19}.
In subsequent work, they considered QFD gates constructed as tensor products of integer powers of the ``$\bT$'' gate, $\bT := \text{diag}\left( 1, \exp\left( \frac{i\pi}{4} \right) \right)$.
There they examined the result of conjugating Pauli matrices by such gates, and fully characterized the Pauli expansion of the (Clifford) result.
Then they used this characterization to understand when such a physical operation preserves the code subspace of a \emph{stabilizer} quantum error-correcting code~\cite{Rengaswamy-arxiv19,Rengaswamy-arxiv20}.
This is fundamental because such codes are necessary to make the quantum computer tolerate noise, and all operations on the encoded information have to be performed by such codespace-preserving physical fault-tolerant operations.
Furthermore, for universal quantum computation we need to implement at least one non-Clifford gate, and the $\bT$ gate is one of the easiest non-Clifford gates to engineer.
As a general recipe, one could replace this (tensor product) gate with any high fidelity lab operation and attempt to repeat this process to understand required code structure.

The Heisenberg-Weyl expansion considered in this paper is intimately connected with the Wigner functions~\cite{Wigner32} via the Fourier transform. The discrete counterparts are explored in~\cite{Wootters04,PhysRevA.70.062101,Gross06}. These insights have played an important role in the simulation and general understanding of magic states, as well as with non-stabilizer resources~\cite{Veitch_2014}.

In this paper, we make contributions towards a few related questions about the hierarchy.
First, surprisingly, the Pauli support (i.e., the Pauli matrices with non-zero coefficients in the Pauli expansion) of even the well-known Clifford group operations remains unknown.
(Note that conjugating a Pauli matrix by the transversal $\bT$ gate produces a Clifford gate, so the aforementioned result already calculates the Pauli expansion of certain types of Cliffords.)
Hence, we study the support of the ``standard'' Clifford operations that correspond to the standard Clifford gate-set consisting of Hadamard, Phase, Controlled-NOT (CNOT), and Controlled-$Z$ (CZ) gates.

Second, Zeng et al.~\cite{ZCC08} considered certain operations called \emph{semi-Clifford} unitaries in the hierarchy~\cite{Gross-qic07}, which have the advantage that they require fewer ancillae than general unitaries in the teleportation model of Gottesman and Chuang~\cite{Zhou-pra00}.
They also showed that any semi-Clifford $\bU$ can be expressed as $\bU = \bG_1 \bD \bG_2$, where $\bD$ is diagonal and $\bG_1, \bG_2$ are Clifford operators.
Cui et al.~\cite{CGK17} have recently characterized the diagonal unitaries in the Clifford hierarchy.
We prove a general result that provides an exact decomposition of any semi-Clifford operation in terms of diagonal gates and physical permutation operators composed of CNOTs and Pauli $X$'s.
Thus, when combined with~\cite{CGK17} and our contribution of characterizing the Pauli support of standard Cliffords, this essentially produces the Weyl expansion of semi-Cliffords.

Third, Zeng et al. conjectured in the above paper that all unitaries in $\cC^{(3)}$ are semi-Clifford and all unitaries in $\cC^{(k)}$ are \emph{generalized} semi-Clifford for any $k$.
While a semi-Clifford operation maps, by conjugation, a maximal commutative subgroup (MCS) of the Pauli group to another MCS of the Paulis, a generalized semi-Clifford operation maps the \emph{span} (i.e., complex linear combination) of a MCS to the span of another MCS.
It is well-known that for $m=1,2$ qubits all unitaries are semi-Clifford, and for $m=3$ qubits the third level is semi-Clifford, so these conjectures are for $k=3$ for all $m > 3$ and for $k \geq 4$ for all $m \geq 3$, respectively. 
Gottesman and Mochon have provided a counterexample for $\cC^{(3)}$ that disproves the semi-Clifford conjecture~\cite{Zeng-phd09}\footnote{The authors of~\cite{RCP19} were unaware of this result, and they regret reporting that this conjecture remained open.}.
Subsequently, Beigi and Shor~\cite{BS08} proved that all unitaries in $\cC^{(3)}$ are generalized semi-Clifford operations, thereby settling the conjecture for $k = 3$.
In this paper, we prove the stronger result that for any unitary $\bC$ from $\cC^{(3)}$, there exists a Clifford $\bG$ such that $\bG \bC$ is supported on a MCS of the Pauli group.
Our proof uses a much simpler induction argument based on the fact that any third level unitary must map (under conjugation) at least one Pauli to some other Pauli.

Finally, the third level of the hierarchy is of particular interest since any third level gate enables universal quantum computation when combined with the Clifford group~\cite{Gottesman-nature99,Boykin-arxiv99}.
When a $\cC^{(3)}$ gate acts by conjugation on a Pauli matrix, the result is a Hermitian Clifford, one example being the aforementioned case of choosing a $\cC^{(3)}$ operation that is a tensor product of integer powers of $\bT$.
It is well-known that Clifford transvections (that is, square roots of Hermitian Pauli matrices; see~\eqref{e-transvection}) form a different generating set for all Cliffords, compared to the standard Clifford gate set mentioned earlier~\cite{Callan,OMeara,Koenig-jmp14}.
We prove a necessary and sufficient condition for the Paulis involved in the transvection decomposition of an arbitrary Hermitian Clifford operator.
Since expanding the product of transvections provides the Pauli expansion of these Hermitian Cliffords, this can potentially be applied to extend the aforementioned result on characterizing stabilizer codes that support transversal $\bT$ gates to other gates from $\cC^{(3)}$.

As a different application, \emph{flag gadgets} have recently become popular as a near-term method to detect correlated faults in circuits~\cite{Chao-arxiv17a,Chao-arxiv17b,Chao-arxiv19,Chamberland-quantum18,Tansuwannont-arxiv18}.
The idea is to introduce a multi-qubit Pauli measurement before and after the circuit, using one or more ancilla qubits, such that the extra gadget acts trivially in the case of no errors but catches catastrophic errors otherwise.
A key requirement to construct flag gadgets for a specific application circuit is to determine the best Pauli measurement to apply before the circuit and identify the result of ``propagating'' the Pauli through the circuit, i.e., determine the result of conjugating the Pauli by the circuit.
The simplest case is to use a Pauli operator that commutes with the circuit. 
For this purpose, any Pauli in the centralizer/dual of the support (of the circuit/unitary) would suffice.
Hence, our aforesaid results on characterizing Pauli supports can be applied to determine the Paulis that commute with the corresponding circuit.
In particular, since flag gadgets are generally applied only to Clifford circuits, our result that any $\cC^{(3)}$ element is supported on a MCS of the Paulis, up to multiplication by a Clifford, provides a way to determine a Pauli that commutes with a non-Clifford element.
Therefore, this insight could be used to design flag gadgets beyond Clifford (subsections of) circuits.



\section{Preliminaries}

\subsection{The Binary Symplectic Group}
We will denote by $\GL(n)$ and $\Sym(n)$ the groups of $n\times n$ invertible and symmetric matrices over the binary field $\F_2$, respectively. 
Addition in $\F_2$ will be denoted by $\oplus$. 
The~\emph{binary symplectic group} $\Sp(2m)\subset \GL(2m)$ is the set of $2m\times 2m$ binary matrices that preserve the~\emph{symplectic inner product} in $\Fm$:
\begin{equation}\label{e-sip}    \inners{(\ba,\bb)}{(\bc,\bd)} = \ba\bd\T \oplus \bb\bc\T = (\ba,\bb)\Omega(\bc,\bd)\T,
\end{equation}
where
\begin{equation}
    \Omega = \fourmat{\textbf{0}_m}{\bI_m}{\bI_m}{\textbf{0}_m},
\end{equation}
A matrix $\bF=\fourmat{\bA}{\bB}{\bC}{\bD}\in\Sp(2m)$ satisfies $\bF\Omega\bF\T = \Omega$, which in turn is equivalent with $\bA\bB\T, \bC\bD\T \in \Sym(m)$ and $\bA\bD\T \oplus \bB\bC\T = \bI_m$.

In $\Sp(2m)$ we distinguish two subgroups:
\begin{align}
    S_D & := \left\{\bF_D(\bP) = \fourmat{\bP}{\textbf{0}_m}{\textbf{0}_m}{\bP^-\T}\,\,\middle| \,\,\bP\in \GL(m)\right\}\\ & \cong \GL(m),\nonumber\\
    S_U & :=  \left\{\bF_U(\bS) = \fourmat{\bI_m}{\bS}{\textbf{0}_m}{\bI_m}\,\,\middle| \,\,\bS\in \Sym(m)\right\} \label{e-S_U}\\ &\cong \Sym(m)\nonumber.
\end{align}

Then every $\bF\in \Sp(2m)$ can be \emph{Bruhat-decomposed} \cite{MR18,PTC20} as
\begin{equation}\label{e-Bruhat1}
    \bF = \bF_D(\bP_1)\bF_U(\bS_1)\bF_{\Omega}(r)\bF_U(\bS_2)\bF_D(\bP_2),
\end{equation}
where 
\begin{equation}
    \bF_{\Omega}(r) = \fourmat{\Imrr}{\Imr}{\Imr}{\Imrr},
\end{equation}
with $\Imr$ being the block matrix with $\bI_r$ in upper left corner and 0 elsewhere, and $\Imrr = \bI_m - \Imr$. 
Here $r = \rk(\bC)$. The semidirect product $S$ of $S_D$ and $S_U$ corresponds to $r = 0$, that is, symplectic matrices with $\bC = {\bf 0}$. 
Let $\widetilde{S}$ be the subgroup of $S$ consisting of matrices $\bF_D(\bP)\bF_U(\bS)$ with $\bP$ upper triangular. Then $\widetilde{S}$ has size $2^{m(m-1)/2} \cdot2^{m(m+1)/2}= 2^{m^2}$ and is a 2-Sylow subgroup of $\Sp(2m)$ that contains $S_U$.

Another type of decomposition of $\Sp(2m)$ can be achieved via~\emph{symplectic transvections} $\bT_\bv := \bI_{2m} + \Omega\bv\T\bv, \, \bv \in \Fm$. 
Such matrix acts on $\Fm$ as $\bx \longmapsto \bx+ \inners{\bv}{\bx}\bv$, and thus $\bT^2_\bv = \bI_{2m}$. 
In general, $\bF\in \Sp(2m)$ is said to be an \emph{involution} if $\bF^2 = \bI_{2m}$ and is said to be \emph{hyperbolic} if $\inners{\bv}{\bv\bF} = 0$ for all $\bv\in \Fm$. It is well-known that symplectic transvections generate $\Sp(2m)$~\cite{Callan,OMeara,Koenig-jmp14}. 
It is shown there that a non-hyperbolic involution can be written as product of $r$ transvections $\bT_{\bv_1},\ldots,\bT_{\bv_r}$, where $r = 2m - \dim(\Fix(\bF)) = \dim(\Res(\bF))$ and 
\begin{align}
    \Fix(\bF) & := \ker(\bI\oplus \bF) := \{\bv\in \Fm\mid \bv = \bv\bF\},\label{e-fix}\\
    \Res(\bF) & := \rs(\bI \oplus \bF) := \{\bv\oplus\bv\bF \mid \bv\in \Fm\}.\label{e-res}
\end{align}
Throughout the paper $\rs(\bullet)$ will denote the row space of a matrix. Note here that, by definition, $\Fix(\bF)$ and $\Res(\bF)$ are dual of each other. The vectors $\bv_1,\ldots,\bv_r \in \Res(\bF)$ must be independent, in which case we say that corresponding transvections are independent. 
On the other hand, a hyperbolic involution can be written as a product of $r+1$ transvections ($r$ as above), $r$ of which are independent and the additional one is dependent of the others.
We will see that the \emph{residue} space $\Res(\bF)$ of a symplectic $\bF$ is intimately connected with the support~\eqref{e-supp} of the corresponding Clifford $\bG$. On the other hand, the \emph{fixed} space $\Fix(\bF)$ being the dual of $\Res(\bF)$ is intimately connected with the Paulis that commute with $\bG$. Transvections are the simplest form of involutions in $\Sp(2m)$ and will play a central role throughout the paper for the simple reason that they correspond to square roots of Hermitian Paulis; see~\eqref{e-transvection} and \eqref{e-ConjTrans}.

\subsection{Quantum Computation}
Fix $N = 2^m$. The standard basis vectors of $\C^N$ will be indexed by binary vectors and denoted as kets, that is, $\be_\bv = |\bv\r,\bv\in\F_2^m$, will have 1 in the position indexed by $\bv$ and 0 else. The \emph{Heisenberg-Weyl} group is defined as
\begin{equation}
    \cH\cW_N := \{i^k\bD(\ba,\bb)\mid \ba,\bb\in \F_2^m, k \in \Z_4\}\subset \U(N),
\end{equation}
where
\begin{equation}
    \bD(\ba,\bb):|\bv\r\longrightarrow (-1)^{\bb\bv\T}|\bv\oplus\ba\r.
\end{equation} 
We will denote by $\cP\cH\cW_N:=\cH\cW_N/\{\pm\bI_N,\pm i\bI_N\}$ the ~\emph{projective} Heisenberg-Weyl group. Directly by definition, we have
\begin{equation}\label{e-Dab}
    \bD(\ba,\bb)\bD(\bc,\bd) = (-1)^{\bb\bc\T}\bD(\ba\oplus\bc,\bb\oplus\bd).
\end{equation}
We will also define $\bE(\ba,\bb):= i^{\ba\bb\T}\bD(\ba,\bb)$, which constitute the Hermitian matrices in $\cH\cW_N$. If follows by~\eqref{e-Dab} that such matrices satisfy
\begin{equation}\label{e-Eab}
    \bE(\ba,\bb)\bE(\bc,\bd) = i^{\bb\bc\T - \ba\bd\T}\bE(\ba+\bc,\bb+\bd),
\end{equation}
where we view all binary vectors as integer vectors and operations are done modulo 4; see ~\cite[Rem.~1]{RCP19} for the meaning of $\bE(\ba,\bb)$ with $(\ba,\bb)\in \Z_4^{2m}$. If the arithmetic of the arguments of operators $\bE(\ba,\bb)$ were to be done modulo 2 one would have
\begin{align}
& \bE(\ba,\bb) \bE(\bc,\bd) \nonumber \\
    & = (-1)^{\scriptsize\inners{(\ba,\bb)}{(\bc,\bd)}} \bE(\bc,\bd) \bE(\ba,\bb) \\
\label{eq:Eab_identity}
    & = i^{\bb \bc\T - \ba \bd\T} \bE(\ba+\bc, \bb+\bd) \\
    & = i^{\bb \bc\T - \ba \bd\T} \bE(\ba+\bc, (\bb \oplus \bd) + 2 (\bb \ast \bd) ) \\
    & = i^{\bb \bc\T - \ba \bd\T} (-1)^{(\ba + \bc) (\bb \ast \bd)\T} \bE( (\ba \oplus \bc) + 2 (\ba \ast \bc), \bb \oplus \bd) \\
    & =  i^{\bb \bc\T - \ba \bd\T} (-1)^{(\ba \oplus \bc) (\bb \ast \bd)\T + (\bb \oplus \bd) (\ba \ast \bc)\T} \bE(\ba \oplus \bc, \bb \oplus \bd).
\end{align}
Above, the asterisk stands for the coordinate-wise product. 
We see that binary arithmetic only ever introduces an additional sign. Thus when the sign is not relevant (e.g.,~\eqref{e-cliff}) we will stick to binary arithmetic.
\begin{rem}\label{R-DEab}
From~\eqref{e-Dab} we have that $\bD(\ba,\bb)$ and $\bD(\bc,\bd)$ commute iff $\inners{(\ba,\bb)}{(\bc,\bd)} = 0$, and otherwise they anticommute.
Similarly,~\eqref{e-Eab} implies $\bE(\ba,\bb)\bE(\bc,\bd) = \pm\bE(\ba+\bc,\bb+\bd)$ if $\inners{(\ba,\bb)}{(\bc,\bd)} = 0$ and $\bE(\ba,\bb)\bE(\bc,\bd) = \pm i\bE(\ba+\bc,\bb+\bd)$ otherwise.
\end{rem}

A~\emph{stabilizer} is a commutative subgroup of $\cH\cW_N$ generated by Hermitian matrices of form $\pm\bE(\ba,\bb)$ that does not contain $-\bI_N$. Thus either $\bE$ or $-\bE$ belong to a stabilizer, but not both.
We will write $S = \bE(\bA,\bB)$ if the stabilizer $S$ is generated by $\bE(\ba_1,\bb_1),\ldots \bE(\ba_k,\bb_k)$, where $\bA$ and $\bB$ are $k\times m$ matrices obtained by stacking $\ba_i$'s and $\bb_i$'s. 
Since $S$ is abelian, the matrix $\bC = (\bA\,\,\bB)$ satisfies $\bC\Omega\bC\T = {\bf 0}$, and thus the row space of $\bC$ is a self-orthogonal (isotropic) subspace of $\Fm$, with respect to the symplectic inner product. A~\emph{maximal stabilizer}, or a ~\emph{maximal commutative subgroup} (MCS) is a stabilizer of size $2^m$. 
Of particular interest are MCSs
\begin{align}
    X_N & = \bE(\bI_m, {\bf 0}_m) = \{\bE(\ba,{\bf 0}) \mid \ba \in \F_2^m\}, \\
    Z_N & = \bE({\bf 0}_m,\bI_m) = \{\bE({\bf 0},\bb) \mid \bb \in \F_2^m\}.
\end{align}
We will refer to their elements as \emph{X stabilizers} and \emph{Z stabilizers}, respectively. Naturally, we identify $X$ stabilizers with vectors $(\ba,{\bf 0}) \in \Fm$ and $Z$ stabilizers with vectors $({\bf 0},\bb)\in \Fm$.

Let $S$ be a stabilizer group of size $2^k$. For $\varepsilon \in \{1,-1\}$, the complex vector space 
\begin{equation}\label{e-qecc}
    \cF_\varepsilon(S) := \{ v \in \C^N\mid \bE v = \varepsilon v \text{ for all } \bE\in S\}
\end{equation}
has dimension $2^{m-k}$. In literature, $\cF_+(S)$ is known as the $[\![m,m-k]\!]$ \emph{stabilizer code associated to $S$}~\cite{NC00}, which encodes $m-k$ \emph{logical} qubits to $m$ \emph{physical} qubits. 
It follows that a MCS $S$ defines a $[\![m,0]\!]$ stabilizer code, that is, $\dim (\cF_\varepsilon(S)) = 1$. 
For this reason $|\psi_\varepsilon\r := \cF_\varepsilon(S)$ 
is called a \emph{stabilizer state}. 
Let $S = \l\bE_1,\ldots,\bE_k\r$ be the stabilizer group generated by the commuting Hermitian Paulis $\{\bE_1,\ldots,\bE_k\}$.  For $\bd\in \F_2^k$, we will denote $S_\bd:=\l (-1)^{d_1}\bE_1,\ldots, (-1)^{d_k}\bE_k\r$. Then 
\begin{equation}
    \Pi_\bd := \prod_{n=1}^k\frac{\bI_N + (-1)^{d_n}\bE_n}{2} = \frac{1}{2^k}\sum_{\bE\in S_\bd}  \bE
\end{equation}
is a projection onto $\cF_+(S_\bd)$, which in turn gives a \emph{resolution} of the identity~\cite[Sec. 10.5]{NC00}:
\begin{equation}
    \sum_{\bd\in \F_2^k}\Pi_\bd = \bI_N.
\end{equation}

\subsection{The Clifford Hierarchy}

The~\emph{Clifford hierarchy} $\{\cC^{(k)},k \geq 1\}$ is defined recursively, where the first level is the Heisenberg-Weyl group, and higher levels are defined by
\begin{equation}
\label{eq:cliff_hierarchy}
    \cC^{(k)} = \{\bU\in\U(N)\mid \bU \cH\cW_N\bU^\dagger \subset \cC^{(k-1)}\}.
\end{equation}
By definition, the~\emph{Clifford group} $\cl_N$ is the second level of the hierarchy up to overall phases, that is $\cl_N := \cC^{(2)}/\U(1)$. 
The following example shows that the Clifford hierarchy does not exhaust $\U(N)$ and also motivates the Weyl expansion.
\begin{exa}
\label{ex:UoutsideCliffHrchy}
Set $\bE_1 = \bE(010,010),\bE_2 = \bE(011,001),\bE_3 = \bE(001,111), \bE_4 = \bE(101,011)$. Then $\bW = (\bE_1+\bE_2+\bE_3+\bE_4)/2$ is easily seen to be outside of $\cl_N$. Further, set $\bE = \bE(100,000)$. We have $\bW\bE\bW^\dagger = \bE_3\bW\bE$ and $(\bE_3\bW\bE)\bE(\bE_3\bW\bE)^\dagger = -\bE\bE_3\bW$. Thus, since multiplication by Paulis (or even Cliffords) preserves the level, iterative conjugation cannot bring $\bE$ up to the same level as $\bW$.
\end{exa}

Let $\{\be_1,\ldots,\be_{2m}\}$ be the standard basis of $\Fm$, and consider $\bG\in \cl_N$. Let $\bc_i\in \Fm$ be such that 
\begin{equation}
    \bG\bE(\be_i)\bG^\dagger = \pm\bE(\bc_i).
\end{equation}
Then the matrix $\bF_{\bG}$ whose $i$th row is $\bc_i$ is a symplectic matrix such that
\begin{equation}\label{e-cliff}
    \bG\bE(\bc)\bG^\dagger = \pm\bE(\bc\bF_{\bG})
\end{equation}
for all $\bc\in \Fm$.
We thus have a group homomorphism 
\begin{equation}\label{e-Phi}
    \Phi : \cl_N\longrightarrow \Sp(2m),\, \bG\longmapsto \bF_{\bG}.
\end{equation}
In addition, $\Phi$ is surjective with kernel $\ker \Phi = \cP\HW$~\cite{RCKP18}, and thus $\cl_N/\cP\cH\cW_N \cong \Sp(2m)$.

Given the decomposition~\eqref{e-Bruhat1}, one is naturally interested on preimages of respective symplectic matrices via $\Phi$. Namely, the unitary matrices
\begin{align}\label{e-GDP}
    \bG_D(\bP) & := |\bv\r\longmapsto |\bv\bP\r,\nonumber\\
    \bG_U(\bS) & := \diag\left(i^{{\bv\bS\bv\T} \mod 4}\right)_{\bv\in\F_2^m}, \\
    \bG_{\Omega}(r) & :=(\bH_2)^{\otimes r}\otimes \bI_{2^{m-r}}\nonumber,
\end{align}
where $\bH_2$ is the $2\times 2$ Hadamard matrix, correspond to $\bF_D(\bP),\bF_U(\bS),$ and $\bF_\Omega(r)$, respectively~\cite{CCKS97}. Strictly speaking, a preimage $\Phi^{-1}(\bF)$ is meant up to $\cH\cW_N$ (and up to a eighth root of unity which we have disregarded throughout~\cite{CRSS98}). 

\begin{rem}
Since $\Phi$ is a homomorphism we have that $\Phi(\bG^\dagger) = \bF_\bG^{-1}$. It follows that if $\bG\in\cl_N$ is Hermitian then $\bF_\bG$ is a symplectic involution. Conversely, if $\bF$ is a symplectic involution then $\bG = \Phi^{-1}(\bF)$ satisfies $\bG^2 \in \cH\cW_N$.
\end{rem}

We will call $\bG\in \cl_N$ a \emph{Clifford transvection} if $\Phi(\bG)$ is a symplectic transvection. For $\bv \in \Fm$ define\footnote{Note that $(\bI_N+ i^\lambda\bE(\bv))/\sqrt{2}$ is unitary iff $\lambda = 1,3$, and otherwise it is a projection.} 
\begin{align}\label{e-transvection}
\bG_\bv:= \frac{\bI_N \pm i\bE(\bv)}{\sqrt{2}}. 
\end{align}
For $\bW \in \cH\cW_N$ we have
\begin{align}\label{e-ConjTrans}
    \bG_\bv\bW\bG_\bv^\dagger = \begin{cases}\bW, & \text{ if } \bW\bE(\bv) = \bE(\bv)\bW, \\ \mp i\bW\bE(\bv), & \text{ if } \bW\bE(\bv) = -\bE(\bv)\bW.\end{cases}
\end{align}
It follows that $\Phi(\bG_\bv) = \bT_\bv$, and any Clifford transvection is of this form (up to $\cH\cW_N$).

In addition, we mentioned that $\Sp(2m)$ is generated by transvections, thus $\bG\in\cl_N$ is a product of Clifford transvections. We have proved the following.
\begin{prop}\label{P-GT}
Any Clifford matrix $\bG\in \cl_N$ can be written as
\begin{equation}\label{e-GT}
    \bG = \bE_0\prod_{n = 1}^k\frac{\bI_N + i\bE_n}{\sqrt{2}} = \frac{\bE_0}{\sqrt{|S|}}\sum_{\bE\in S}\alpha_\bE\bE,
\end{equation}
where $\bE_0\in \cH\cW_N, S = \l\bE_1,\ldots,\bE_k\r$, and $\alpha_\bE \in \C$.
\end{prop}

One of the goals of this paper is to determine the $\alpha_\bE$'s produced by Clifford matrices. In particular we will see that if $\bE_1,\ldots,\bE_k$ are independent then $\alpha_\bE \in \{\pm1,\pm i\}$; see~\eqref{e-G}-\eqref{e-G1}. If $\bG$ is Hermitian then $\bF_\bG$ is an involution, and thus $k\in \{r,r+1\}$ where $r = 2m - \dim(\Fix(\bF_\bG))$ with at least $r$ transvections being independent.


\section{Support of the Clifford Group}\label{sec-support}

The set $\cE_N = \{\bE(\bc) \mid \bc\in\Fm\}$ is an orthonormal basis for the vector space $\cM_N(\C)$ of $N\times N$ complex matrices 
with respect to the Hermitian inner product
\begin{equation}
    \inner{\bM}{\bN}:= \frac{1}{N}\Tr(\bM^\dagger\bN).
\end{equation}
Thus any matrix $\bM\in \cM_N(\C)$ is a linear combination
\begin{equation}\label{e-sum}
    \bM = \sum_{\bc\in\Fm}\alpha_{\bc} \bE(\bc), \quad \alpha_\bc = \inner{\bE(\bc)}{\bM}\in \C.
\end{equation}
We are interested in sums of Pauli matrices that yield Clifford matrices. The \emph{support} of $\bM\in \cM_N(\C)$ as in~\eqref{e-sum} with respect to $\cE_N$ is defined as 
\begin{align}\label{e-supp}
    \supp(\bM)  & := \{\bE(\bc)\in\HW\mid \alpha_\bc \neq 0\} \\ & \cong \{\bc\in\Fm\mid \alpha_\bc \neq 0\}.
\end{align}
When dealing with the support, we will conveniently switch between the two equivalent definitions. Thus $\bE(\bc)\in \supp(\bM)$ iff $\Tr(\bM^\dagger\bE(\bc))\neq 0$. We say in this case that $\bM$ is \emph{supported} on $\supp(\bM)$.
\begin{rem}
\begin{arabiclist}
\item Since $\bE(\bc)$ differs from $\bD(\bc)$ only by a factor $i^k$ we see that $\Tr(\bM^\dagger\bE(\bc)) \neq 0$ iff $\Tr(\bM^\dagger\bD(\bc)) \neq 0$. Thus, to avoid the additional scaling factor, we will use matrices $\bD(\bc)$ when computing supports/traces.
\item It follows directly by the definition of support and~\eqref{e-Dab} that the support of $\bM\bD(\bx)$ (or $\bD(\bx)\bM$) is just the translate $\{\bx\} + \supp(\bM)$ for all $\bx\in \Fm$.
\end{arabiclist}
\end{rem}

It is clear that $\cH\cW_N$ is supported on singletons. It follows that a unitary $\bM$ is a Clifford matrix iff $\bM \bE\bM^\dagger$ is supported on a singleton for all $\bE\in \cH\cW_N$. On the other hand, for $\bG\in \cl_N$, we have $|\supp(\bM)| = 2$ iff $\bG$ is a Clifford transvection (up to $\cH\cW_N$). In general, we have the following immediate consequence of Proposition~\ref{P-GT}.
\begin{cor}\label{cor}
Any Clifford matrix $\bG$ is supported either on a group $S$ or on a coset $\bE_0S$ depending on whether $\bG$ has trace or not.
\end{cor}
\begin{proof}
Observe 
that $\bE\in\cH\cW_N$ is traceless unless $\bE = \bI_N$. Thus, for $\bG$ as in~\eqref{e-GT} we have that $\Tr(\bG)\neq 0$ iff $\bE_0 \in S$. It also follows that $\supp(\bG) = \bE_0S$.
\end{proof}
\begin{rem}
One can easily construct non-Clifford matrices supported on a subgroup. For instance, it follows easily from Proposition~\ref{P-lc} that $\bT^{\otimes m}$ is supported on the subspace $\{{\bf 0}\}\times \F_2^m$, or alternatively, on the subgroup $Z_N$ of diagonal Paulis. Thus, Corollary~\ref{cor} does not completely characterize $\cl_N$. In fact, for any $\bG\in \cl_N$, we have 
\begin{align}
    \bE(\bc) \in \supp(\bG) & \iff \Tr\big((\bG\bE(\bc)\bG^\dagger)\bG\big) \neq 0 \\
    & \iff \Tr\big(\bE(\bc\bF_\bG)\bG\big) \neq 0\\
    & \iff \bE(\bc\bF_\bG) \in \supp(\bG),
\end{align}
which implies that the support of $\bG$ is an~\emph{invariant} subspace of $\bF_\bG$.
\end{rem}

By definition, in order to understand $\bG\in \cl_N$ it is sufficient to understand its action on $\cH\cW_N$, which thanks to~\eqref{e-Phi} can be understood via the action of the corresponding symplectic matrix $\bF_\bG$ on $\Fm$. Consider the fixed space $\Fix(\bF_\bG)$ from~\eqref{e-fix}. Then $\bc \in \Fix(\bF_\bG)$ iff 
\begin{equation}
    \bG\bE(\bc)\bG^\dagger = \pm \bE(\bc\bF_\bG) = \pm \bE(\bc),
\end{equation}
that is, iff $\bE(\bc)$ either commutes or anticommutes with $\bG$. Let us now consider the Pauli matrices that commute with $\bG$, that is
\begin{equation}
    C_\bG = \{\bc\in \Fix(\bF_\bG)\mid \bG\bE(\bc)\bG^\dagger = \bE(\bc)\}.
\end{equation}
It is then clear that the quotient $\Fix(\bF_\bG)/C_\bG$ captures the Pauli matrices that anticommute with $\bG$. We will denote by $(\sbt)^\sperp$ the dual w.r.t. the symplectic inner product~\eqref{e-sip}. With this notation we have the following.

\begin{prop}\label{P-supp}
$\supp(\bG) \subseteq C_\bG^\sperp$. 
\end{prop}
\begin{proof}
We will show the reverse inclusion of the complements. Indeed, let $\bc \in \Fm$ be such that $\bc \notin C_\bG^\sperp$. Then, there exists $\bv \in C_\bG$ such that $\inners{\bc}{\bv} = 1$. It follows that $\bE(\bv)$ commutes with $\bG$ and anticommutes with $\bE(\bc)$. Thus
\begin{align}
    \Tr\big(\bG\bE(\bc)\big) & = \Tr\Big(\bE(\bv)^\dagger\bE(\bv)\bG\bE(\bc)\Big) \\ 
    & = \Tr\Big(\bE(\bv)^\dagger\bG\bE(\bv)\bE(\bc)\Big) \\
    & = -\Tr\Big(\bE(\bv)^\dagger\bG\bE(\bc)\bE(\bv)\Big) \\
    & = -\Tr\big(\bG\bE(\bc)\big),
\end{align}
which in turn implies $\Tr(\bG\bE(\bc)) = 0$, and hence $\bc\notin \supp(\bG)$.
\end{proof}


Next, we completely characterize the supports of standard Clifford matrices~\eqref{e-GDP} in terms of the invariants~\eqref{e-fix}-\eqref{e-res} of the defining symplectic matrices.
\begin{prop}\label{P-lem}
The support of standard Clifford matrices introduced in~\eqref{e-GDP} satisfies the following:
\begin{arabiclist} 
\item $\supp(\bG_D(\bP)) = \Res(\bP^{-1})\times \Fix(\bP)^\perp = \Res(\bP^{-1})\times \Res(\bP)$.
\item Let $\bS \in \Sym(m)$ and $W = \ker(\bS) = \{\bw\in \F_2^m\mid \bw\bS = {\bf 0}\}$. If $\Tr(\bG_U(\bS)) \neq 0$ then $\supp(\bG_U(\bS)) = \{{\bf 0}\}\times W^\perp$. 
Otherwise $\bG_U(\bS)$ is supported on a coset of $\{{\bf 0}\}\times W^\perp$. As a consequence, the support of diagonal Cliffords is completely characterized by the row/column space of the associated symmetric $\bS$.
\item Let $D_r=\{(\bx,{\bf 0}_{m-r},\bx,{\bf 0}_{m-r})\mid \bx \in \F_2^r\}\subset \Fm$. Then $\supp(\bG_\Omega(r)) = ({\bf 1}_r,{\bf 0}_{2m-r})\oplus D_r$, where ${\bf 1}_r$ denotes the all ones vector of size $r$.
As a consequence, partial Hadamard matrices $\bG_\Omega(r)$ are supported on a coset of $\Res(\bF_\Omega(r))$.
\end{arabiclist}
\end{prop}

\begin{proof}
(1) By definition we have $\bG_D(\bP) = \sum_{\bv\in\F_2^m}|\bv\bP\r\l\bv|$. Thus, for $\bG = \bG_D(\bP)$ we have
\begin{align}
    \bG\bD(\ba,\bb) = & \sum_{\bv\in\F_2^m}|\bv\bP\r\l\bv|\sum_{\bw\in\F_2^m}(-1)^{\bw\bb\T}|\bw\oplus\ba\r\l\bw|\\
    = & \sum_{\bv\in\F_2^m}(-1)^{\bv\bb\T}|(\bv\oplus\ba)\bP\r\l\bv|.
\end{align}
For $\ba \in \F_2^m$ we will denote $\Fix_\ba(\bP) = \{\bv\in\F_2^m \mid \bv\oplus\bv\bP = \ba\bP\}$. With this notation we have 
\begin{equation}
    \Tr\big(\bG\bD(\ba,\bb)\big) = \sum_{\bv\in \Fix_\ba(\bP)}(-1)^{\bv\bb\T}.
\end{equation}
Since $\Fix(\bP)$ is a subspace of $\F_2^m$ we have that $\bD(\textbf{0},\bb)\in \supp(\bG)$ iff $\bb \in \Fix(\bP)^\perp$. On the other hand, for $\bx \in \Fix_\ba(\bP)$ we have $\bx \oplus \Fix(\bP)=\Fix_\ba(\bP)$. Indeed, the forward containment is trivial and equality follows due to equal cardinalities. Thus, if $\Fix_\ba(\bP)\neq \emptyset$, we have that $\bD(\ba,\bb)\in \supp(\bG)$ iff $\bb\in \Fix(\bP)^\perp$. Next, recall the subspace $\Res(\bP)$ from~\eqref{e-res}. We have that $\Fix_\ba(\bP)\neq \emptyset$ iff $\ba \in \Res(\bP^{-1})$. We conclude that $\bG_D(\bP)$ is supported on $\Res(\bP^{-1})\times \Fix(\bP)^\perp \subset \Fm$. 
Then by definition $\Fix(\bP)^\perp = \Res(\bP)$.

(2) Let $\bG:=\bG_U(\bS)$. Then
\begin{equation}\label{e-TrFUS}    
\Tr\big(\bG\bD(\ba,\bb)\big) = \sum_{\bv\in\F_2^m}i^{(\bv\oplus\ba)\bS(\bv\oplus\ba)\T + 2\bv\bb\T}\l\bv|\bv\oplus\ba\r.
\end{equation}
It follows that $\bD(\ba,\bb)\in \supp(\bG)$ only if $\ba = {\bf 0}$. So from now on we fix $\ba={\bf 0}$. It is shown in~\cite[Appendix~A]{CHJ10} that the sum in~\eqref{e-TrFUS} is nonzero iff 
\begin{equation}\label{e-TrFUS1}
    \sum_{\bw\in W}i^{\bw\bS\bw\T + 2\bw\bb\T} \neq 0.
\end{equation}
Consider the maps $\chi_\bS:\bw\longmapsto \bw\bS\bw\T$ and $\chi_\bb:\bw\longmapsto + 2\bw\bb\T$, and put $\chi_{\bS,\bb} := \chi_\bS + \chi_\bb$. For $\bv,\bw\in W$ we have
\begin{align}
    \chi_{\bS,\bb}(\bv\oplus\bw) & = (\bv \oplus \bw) \bS (\bv \oplus \bw)\T + 2 (\bv \oplus \bw) \bb\T \\
    & = (\bv + \bw) \bS (\bv + \bw)\T + 2 (\bv + \bw) \bb\T \bmod 4\\
    & = \chi_{\bS,\bb}(\bv) + \chi_{\bS,\bb}(\bw) + 2\bw\bS\bv\T \\
    & = \chi_{\bS,\bb}(\bv) + \chi_{\bS,\bb}(\bw) \mod 4,
\end{align}
where the last equality follows by the fact that $\bw\bS = {\bf 0}\bmod 2$. Thus the map $\bw \longmapsto i^{\chi_{\bS,\bb}(\bw)}$ is a character of $W$. It follows that $\bD({\bf 0},\bb)\in \supp(\bG)$ iff $\chi_{\bS,\bb}(\bw) = 0$ for all $\bw \in W$.

By the above argument, it also follows that $\Tr(\bG)\neq 0$ iff $\chi_\bS(\bw) = 0$ for all $\bw \in W$. 
In this case,~\eqref{e-TrFUS1} reduces to $\sum_{\bw\in W}(-1)^{\bw\bb\T} \neq 0$, which holds iff $\bb\in W^\perp$. 
Similarly, $\Tr(\bG) = 0$ iff $\chi_\bS$ is not the trivial map. 
Since $\chi_\bS$ is an even-valued linear map (mod 4) then there exists $\bc$ (depending on $\bS$) such that $\chi_\bS(\bw) = 2 \bw \bc\T$. But then it is clear that $ \chi_{\bS,\bb}$ is the zero map iff $\bc\oplus\bb \in W^\perp$ iff $\bb \in \bc\oplus W^\perp$.
It follows that $\supp(\bG) = \{{\bf 0}\}\times (\bc\oplus W^\perp )$ for any such $\bc$ as above.

(3) Consider $\bG = \bG_\Omega(r)$ for $0\leq r \leq m$, and let us first handle $r = m$, which corresponds to the fully occupied Hadamard matrix in $N = 2^m$ dimensions. In this case, we have
\begin{align}
     & \bG\bD(\ba,\bb) \nonumber\\ & = \frac{1}{\sqrt{N}}\sum_{\bv,\bw}(-1)^{\bv\bw\T}|\bv\r\l\bw|\sum_{\bx\in\F_2^m}(-1)^{\bx\bb\T}|\bx\oplus\ba\r\l\bx| \\
    & = \frac{1}{\sqrt{N}}\sum_{\bv,\bx\in\F_2^m}(-1)^{\bv(\bx\oplus\ba)\T + \bx\bb\T}|\bv\r\l\bx|.
\end{align}
The above yields
\begin{equation}
    \Tr\big(\bG\bD(\ba,\bb)\big) = \frac{1}{\sqrt{N}}\sum_{\bv\in \F_2^m}(-1)^{\bv(\bv\oplus\ba\oplus\bb)\T}.
\end{equation}
Then map $\bv\longmapsto \bv(\bv\oplus\ba\oplus\bb)\T$ is additive, and it is the trivial map iff $\ba\oplus\bb = {\bf 1}_m$. Thus $\bD(\ba,\bb)\in \supp(\bG)$ iff $\ba+\bb = {\bf 1}_m$. 
It follows that $\supp(\bG) = ({\bf 1}_m,{\bf 0}_m)\oplus D_m$. Then, for $r < m$, a similar argument implies that  $\bD(\ba,\bb)\in \supp(\bG)$ iff $\ba_{1:r} + \bb_{1:r} = {\bf 1}_r$ and $\ba_{r+1:m} = \bb_{r+1:m} = {\bf 0}_{m-r}$. The proof is concluded with the observation that $D_r = \Res(\bF_\Omega(r))$.
\end{proof}

\begin{rem}
Let $\bP\in \GL(m)$ and $\bS \in \Sym(m)$, and put $\bG = \bG_D(\bP)\bG_U(\bS)$. Then
$\bG = \sum_{\bv\in \F_2^m} i^{\bv\bS\bv\T}|\bv\bP\r\l\bv|$, which in turn yields,
\begin{align}
    \Tr\big(\bG\bD(\ba,\bb)\big) & = \Tr\left(\sum_{\bv\in \F_2^m} i^{\bv\bS\bv\T + 2\bv\bb\T}|(\bv\oplus\ba)\bP\r\l\bv|\right)\\
    & = \sum_{\bv\in \Fix_\ba(\bP)}i^{\bv\bS\bv\T + 2\bv\bb\T}.
\end{align}
Now the analysis continues as in Proposition~\ref{P-lem}(2); see also~\cite[Lem.~6]{CHJK10} for further details.
\end{rem}




\begin{exa}\label{ex-CNOT}
We saw from Proposition~\ref{P-lem} that matrices $\bG_D(\bP)$ are supported on a subspace/subgroup. This is of course consistent with Corollary~\ref{cor} since these matrices always have trace. Indeed, entry (1,1) is always 1 since ${\bf 0}\bP = {\bf 0}$. The CNOT gate is of form $\bG_D(\bP)$ where 
\begin{equation}
    \bP = \bP^{-1} = \fourmat{1}{1}{0}{1}.
\end{equation}
We have $\Res(\bP^{-1}) = \{00,01\}$ and $\Fix(\bP)^\perp = \{00,10\}$. Thus, $\supp({\rm CNOT}) = \{0000,0010,0100,0110\}$, as one can directly verify. Here we have $m = 2$ and CNOT corresponds to the symplectic $\bF_{{\rm CNOT}} = \bF_D(\bP)$, which is a hyperbolic involution with $\Fix(\bF_{\rm CNOT}) = \supp({\rm CNOT})$. Note also that $\Fix(\bF_{\rm CNOT}) = C_{\rm CNOT}$, and thus equality in Proposition~\ref{P-supp} can be achieved.  In addition $\bF_{{\rm CNOT}} = \bT_{0010}\bT_{0100}\bT_{0110}$. On the complex domain we have
\begin{align}\label{e-CNOT}
    \text{CNOT} & = \frac{1 - i}{\sqrt{2}}\cdot\frac{(\bI + i\bE_1)(\bI + i\bE_2)(\bI -i\bE_1\bE_2)}{\sqrt{8}}\\
    & = \frac{1}{2}(\bI+\bE_1+\bE_2 - \bE_1\bE_2),
    \end{align}
where $\bE_1 = \bE(00,10), \bE_2 = \bE(01,00)$.
It follows that CNOT is supported on the MCS generated by $\bE_1$ and $\bE_2$.
\end{exa}


\begin{exa}
\begin{arabiclist}
\item Let $\bb\in \F_2^m$ and consider the symmetric matrix $\bS_\bb := \bb\T\bb$. In this case $[\ker(\bS_\bb)]^\perp = \{{\bf 0},\bb\}$. In addition $\bG_U(\bS_\bb) = (\bI_N \pm i\bE({\bf 0},\bb))/\sqrt{2}$.
\item Let $\bS$ be a diagonal matrix with diagonal $d_\bS$. Let $r = \hwt(d_\bS)$ be the number of non-zero elements in $d_\bS$. In this case $\bF_U(\bS)$ is a product of transvections $\bT_{\bv_n}$ where $\bv_n = ({\bf 0},\bb_n)$ where $\bb_n$ is the $n$th nonzero row of $\bS$. Then 
\begin{equation}
    \bG_U(\bS) = \prod_{n = 1}^r\frac{\bI_N + i\bE(\bv_n)}{\sqrt{2}},
\end{equation}
from which we may also conclude that $\Tr(\bG_U(\bS)) \neq 0$.
\end{arabiclist}
\end{exa}


We end this section by computing the supports of the \emph{local} Clifford group $(\cl_2)^{\otimes m}\subset \cl_N$. 

\begin{prop}[\mbox{Support of local Cliffords}]\label{P-lc}
Let $\bG = \bG_1\otimes\cdots \otimes \bG_m\in (\cl_2)^{\otimes m}$, and let $S_i$ be the support of $\bG_i$ in $\F_2^2$. Then $\supp(\bG) = \sigma(S_1\times\cdots\times S_m)$, where $\sigma$ is the permutation $(a_1,b_1,\ldots,a_m,b_m)\longmapsto (a_1,\ldots,a_m,b_1,\ldots,b_m)$.
\end{prop}
\begin{proof}
The result follows immediately by the fact that the trace function is multiplicative on pure tensors.
\end{proof}


\section{On Hermitian Clifford matrices}

Hermitian Clifford matrices, on top of being interesting on their own right, they also play a prominent role on understanding the third level of the Clifford hierarchy $\cC^{(3)}$. Indeed, by definition, $\bC\cH\cW_N\bC^\dagger\subset \cl_N$ for all $\bC\in \cC^{(3)}$. 
The conjugate action preserves traces and the Hermitian property. Thus, other than the identity, only traceless Clifford matrices can emerge from conjugate action with a third level matrix. 
With the same notation as in Proposition~\ref{P-GT}, we see that $\bG$ is traceless iff $\bE_0\notin S$. 
Further, $\bC\bE\bC^\dagger$ must also be a Hermitian Clifford matrix for any Hermitian Pauli matrix $\bE$. The corresponding symplectic matrix $\Phi(\bC\bE\bC^\dagger)$ is an involution, and symplectic matrices emerging in this way ($\bC$ fixed, $\bE$ varies) must commute. This fact, although elementary, is crucial because any group of commuting involutions is conjugate\footnote{Two groups $S$ and $S'$ are called conjugate if there exists $g$ such that $S' = gSg^\dagger$.}  with some subgroup of $S_U$ from~\eqref{e-S_U}. 
This means that, for instance, the CNOT gate (or any Clifford matrix of form $\bG_D(\bP)$) cannot emerge from a third level action, despite $\bF_{\rm CNOT}$ being a symplectic involution; see also Example~\ref{ex-CNOT}. 

We have mentioned that any Clifford matrix, is, up to a multiplication by a Pauli, a product of transvections. We have the following structural results if the transvections involved are independent.


\begin{theo}\label{T-T}
Let $\bE_n = \bE(\bc_n),  n = 1,\ldots, k$, be a set of $k$ \emph{independent} Hermitian Pauli matrices. Let also $\bE_0 = \bE(\bc_0)$ be a Hermitian Pauli matrix. Then, the Clifford matrix 
\begin{equation}\label{e-eq}
    \bG = \bE_0\prod_{n=1}^k\frac{1}{\sqrt{2}}(\bI+i\bE_n),
\end{equation}
is Hermitian iff $\bE_0$ anticommutes with all $\bE_n$ and all $\bE_n$ commute with each other. As a consequence, if $\bG$ is Hermitian then it is also traceless.
\end{theo}
\begin{proof}
Let $\bC := (\bA \,\, \bB)$ be the $k\times 2m$ binary matrix whose $n$th row is $\bc_n = (\ba_n,\bb_n)$. Since all the $\bE_n$ are independent we have that $\rk(\bC) = k$.
Using~\eqref{e-Eab} we have that
\begin{equation}\label{e-G}
\begin{split}
    \bG  & =  \frac{\bE_0}{\sqrt{2^k}}\sum_{\bd\in\F_2^k}i^{\bd(\bI_k+\widetilde{\bC})\bd\T}\bE(\bd\bC \!\!\!\!\mod 4)\\ & = \frac{1}{\sqrt{2^k}}\sum_{\bd\in\F_2^k}i^{\bd(\bI_k+\widetilde{\bC})\bd\T}\bE_0\bE(\bd\bC \!\!\!\!\mod 4),
    \end{split}
\end{equation}
where $\widetilde{\bC}$ is the $k\times k$ matrix whose $(i,j)$ entry is $\ba_j\bb_i\!\!\T - \ba_i\bb_j\!\!\T\!\!\mod 4$ if $i<j$ and 0 else. 
If we write $\bc_0 = (\ba_0,\bb_0)$, then~\eqref{e-G} can be further rewritten as 
\begin{equation}\label{e-G1}
    \bG = \frac{1}{\sqrt{2^k}}\sum_{\bd\in\F_2^k}i^{\bd(\bI_k+\widetilde{\bC})\bd\T}i^{\bd\bA\bb_0\!\!\!\T - \ba_0\bB\T\bd\T}\bE(\bd\bC+\bc_0)
\end{equation}
It follows immediately that $\bG$ is Hermitian iff all coefficients in~\eqref{e-G1} are $\pm 1$.

By looking at the standard basis of $\F_2^k$, i.e., $\bd = \be_n$, and corresponding coefficients, we see that $\ba_n\bb_0\!\!\!\T \oplus\ba_0\bb_n\!\!\!\T = 1$
for $n = 1,\ldots,k$, which in turn means that $\bE_0$ anticommutes with $\bE_n$. 
To show that $\bE_j,\bE_n$ commute, consider $\bd$ of weight two with ones in positions $j,n$. 
This corresponds to looking at the coefficient of $\bE_0\bE_j\bE_n$. 
Because $\bE_0$ anticommutes with both $\bE_j,\bE_n$, the term $i^{\bd\bA\bb_0\!\!\!\T - \ba_0\bB\T\bd\T}$ contributes $\pm 1$ and the term $i^{\bd\bI_k\bd\T} = i^{\rm wt(\bd)}$ contributes $-1$. Thus, the overall coefficient will be $\pm 1$ only if $\bE_j,\bE_n$ commute.

For the converse, one could argue similarly to show that the coefficients in~\eqref{e-G1} are $\pm 1$. However, we point out here that the statement follows immediately from~\eqref{e-eq}.

Finally, for Hermitian $\bG$ as in~\eqref{e-eq} we argued that $\bE_0$ anticommutes with all $\bE_n$. Thus $\bE_0$ cannot be contained on the commutative group generated by all $\bE_n$. As we have mentioned earlier, this implies that $\bG$ is traceless.
\end{proof}

\begin{rem}
Recall Proposition~\ref{P-GT} where a Clifford matrix $\bG$ is written as a generic sum of Hermitian matrices $\bE$. In~\eqref{e-G1} we have explicitly computed the coefficients $\alpha_\bE$, and evidently, they are of form $i^{Q(\bv)}, \bv\in \F_2^m$ where $Q$ is a quadratic form $\mod 4$. This generalizes a result of~\cite{CCKS97} (see also~\cite{CHJ10}) where the authors showed that the coefficients of~\emph{diagonal} Clifford matrices are determined by a quadratic form. Indeed, diagonal Clifford matrices are of form $\bG_U(\bS), \bS\in\Sym(m)$, and we see the aforementioned quadratic form in~\eqref{e-TrFUS}-\!\!~\eqref{e-TrFUS1}.
\end{rem}


\begin{rem}
If $\bG \in \cl_N$ is Hermitian, we mentioned that the corresponding symplectic matrix $\bF_\bG$ must be an involution, which we also mentioned can be written as a product of $k \in \{r,r+1\}$ transvections, $r = 2m - \dim(\Fix(\bF_\bG))$, where at least $r$ are independent.
Theorem~\ref{T-T} settles the scenario when $\bF_\bG$ is a product of only independent transvections. 
When the additional transvection $\bT_{r+1}$ is dependent of the other $r$ transvections, then multiplying $\bG = \bT_1\cdots\bT_r$ with $\bT_{r+1}$ may preserve the support of $\bG$ (see Example~\ref{ex-CNOT}) or reduce the support of $\bG$. In the latter instance the supported is reduced by half. Keeping track of the support of $\bG\bT_{r+1}$ becomes tedious and involves sign chasing that depends on the commutativity relation of $\bT_1,\ldots,\bT_r$.
\end{rem}

\section{On (Generalized) Semi-Clifford Matrices}

For $k\geq 3$ the levels $\cC^{(k)}$ of the Clifford hierarchy do not form a group, and thus a complete characterization becomes challenging.
In~\cite{ZCC08} the authors use the notion of semi-Clifford matrices to achieve partial results. 
A unitary matrix $\bU \in \U(N)$ is called \emph{semi-Clifford} if there exists a MCS $S_1\subset\HW$ such that $S_2 = \bU S_1 \bU^\dagger$ is also MCS. 
Since the Clifford group $\cl_N$ permutes stabilizers of a given dimension, a Clifford matrix is trivially semi-Clifford. 
It is shown in~\cite{ZCC08} that for $m = 1,2$ qubits the Clifford hierarchy is comprised of semi-Clifford matrices, and for $m= 3$ qubits the third level $\cC^{(3)}$ is comprised of semi-Clifford matrices. 
Moreover, they show that for $m>2$ qubits that there exist non semi-Clifford matrices in each level $\cC^{(k)},k>3$, and conjecture that the third level $\cC^{(3)}$ is comprised of semi-Cliffords for any number of qubits. The conjecture was disproved by Gottesman and Mochon via a counterexample with $m = 7$ qubits; see~\cite{BS08}. 
On the other hand, the diagonal elements of each level, denoted $\cC^{(k)}_d$ do form a group~\cite[Prop.~4]{ZCC08}, and are completely characterized in~\cite{CGK17}. 
The QFD gates of~\cite{RCP19} represent all 1-  and 2-local diagonal gates in the hierarchy, and thus forming a particularly nice subclass of diagonal gates. 


\begin{rem}\label{R-cleq}
Multiplying by Clifford matrix preserves the levels of the hierarchy. Thus, without loss of generality, we will consider semi-Clifford matrices (and any matrix in the hierarchy) up to multiplication by Clifford matrices. This enables us to adjust any semi-Clifford $\bU$ matrix so that it fixes any given given MCS $S$. Indeed, assume $\bU S_1 \bU^\dagger = S_2$. Let $\bG_1,\bG_2\in \cl_N$ be such that $\bG_1 S \bG_1^\dagger = S_1$ and $\bG_2S_2\bG_2^\dagger = S$. Then $\bG_2\bU\bG_1$ is a semi-Clifford matrix that fixes $S$. As mentioned earlier, by~\cite[Alg. 1]{RCKP18}, there exists $\bG_3\in \cl_N$ such that $\bG_3\bG_2\bU\bG_1$ fixes $S$ pointwise.
\end{rem}

\begin{theo}\label{T-SC}
    Let $\bC\in \cC^{(k)}$ be a unitary matrix that fixes the group of diagonal Paulis $Z_N = \bE({\bf 0}_m, \bI_m)$. Then $\bC = \bD\bE(\ba,{\bf 0})\bG_D(\bP)$, for some diagonal $\bD \in \cC_d^{(k)}, \bP\in\GL(m)$, and $\ba \in\F_2^m$.
\end{theo}

We will make use of the structure of first order Reed-Muller codes to prove Theorem~\ref{T-SC}. Let $C = \{(\bv\bb\T\!\!\mod 2)_{\bv\in\F_2^m}\mid \bb\in\F_2^m\}$. Then, the \emph{first-order Reed-Muller} code is the linear $[2^m,m+1,2^{m-1}]_2$-code
\begin{equation}
    \reed(1,m) = C\cup\{\bc \oplus {\bf 1}\mid \bc\in C\}.
\end{equation}
The automorphism group of $\reed(1,m)$ is the general affine group $\GA(m)$ of maps $\bv\longmapsto \bv\bP \oplus \ba, \bP\in \GL(m),\ba\in\F_2^m$; see~\cite[Chapter~13]{MWS77} for instance.
\vspace{.08 in}

\textit{Proof of Theorem~\ref{T-SC}.} Assume that $\bC\in \cC^{(k)}$ fixes $Z_N$, and let $v\in \C^N$ be a common eigenvector of all $\bD({\bf 0},\bb)\in Z_N$. 
Then, by assumption, we have
\begin{equation}
\bE({\bf 0},\bb)\bC v = \bC\bE({\bf 0},\bb')v = \pm \bC v,  
\end{equation}
which in turn implies that $\bC v$ is also a common eigenvector of $Z_N$. Thus $\bC$ maps the common eigenvector $|\bv\r$ to another common eigenvector which is of the form $\alpha_\bv|\pi(\bv)\r$ for some $\pi(\bv)\in \F_2^m$ and $\alpha_\bv\in\C$. 
In other words, $\bC$ is a monomial map, that is $\bC = \bD\Pi$, where $\bD$ is the diagonal matrix with entry $\alpha_\bv$ in position $\pi(\bv)$ and $\Pi$ is the permutation $\bv \longmapsto \pi(\bv)$. The assumption $\bC\in \cC^{(k)}$ implies $\bD \in \cC^{(k)}_d$. 
By construction, the diagonals of Paulis in $Z_N$ are of form $\pm((-1)^{\bv\bb\T})_{\bv\in\F_2^m}$, and we point out that the exponents of such diagonals are precisely the elements of $\reed(1,m)$. Thus $\Pi$ induces an isometry on $\reed(1,m)$, which as we mentioned must be an invertible affine map. That is,  $\Pi = \bE(\ba,{\bf 0})\bG_D(\bP)$ for some diagonal $\bP\in\GL(m)$, and $\ba \in\F_2^m$. In particular we have $\Pi\in \cl_N$, which along with the assumption $\bC \in \cC^{(k)}$ implies $\bD \in \cC_d^{(k)}$.
 \hfill \qedsymbol

\begin{rem}
In~\cite{ZCC08} it was shown that a semi-Clifford $\bC\in \cC_d^{(k)}$ is of the form $\bC = \bG_1\bD\bG_2$ for some $\bG_1,\bG_2\in \cl_N$ and $\bD\in \cC_d^{(k)}$. 
Theorem~\ref{T-SC} further extends this result by characterizing the Clifford matrices that appear into decomposition of $\bC$. 
Thus, we obtain a complete characterization of semi-Clifford elements in the Clifford hierarchy. We believe that this result, along with the notion of the support can be used in many applications, e.g., design of flag gadgets.
\end{rem}

The argument of Theorem~\ref{T-SC} holds in a slightly more general setting.

\begin{rem}
\begin{arabiclist}
 \item Let $\bC$ be any unitary matrix that fixes a MCS $S = \l\bE_1,\ldots,\bE_m\r$.
For $\bd\in \F_2^m$ denote $S_\bd:=\l (-1)^{d_1}\bE_1,\ldots, (-1)^{d_m}\bE_m\r$, and put $\cF_+(S_\bd) := |\psi_\bd\r$. For any $\bE \in S_\bd$ we have 
\begin{equation}
    \bE\bC|\psi_\bd\r = \bC\bE'|\psi_\bd\r = \pm \bC|\psi_\bd\r,
\end{equation}
and thus $\bC|\psi_\bd\r\in \cF_+(S_{\bd'})$ for some $\bd' \in \F_2^m$. This means that $\bC|\psi_\bd\r = \lambda_\bd|\psi_{\bd'}\r$ for some eigenvalue $\lambda_\bd$. In particular, $\bC$ is a monomial matrix with respect to the eigenbasis $\cE_S:= \{|\psi_\bd\r\mid \bd \in \F_2^m\}$.
\item Let $\bC$ be a unitary matrix that fixes the \emph{span} of $Z_N$, that is, $\bC$ maps any diagonal to another diagonal. Then $\bC$ is a monomial matrix. In particular, any semi-Clifford matrix is a monomial matrix up to some Clifford correction; see also~\cite[Prop.~2]{ZCC08}.
\end{arabiclist}
\end{rem}

Let $\bC \in \cC^{(3)}$ be such that it fixes some subgroup $\widehat{S}$ (of $\cH\cW_N$) under conjugation; see also Remark~\ref{R-cleq}. Then, by~\cite[Alg. 1]{RCKP18}, there exists $\bG \in \cl_N$, produced as a sequence of transvections, such that $\widehat{\bC}: = \bG\bC$ fixes $\widehat{S}$ \emph{point-wise}. Let $\widehat{S}^\sperp$ be all the Pauli matrices that commute with elements of $\widehat{S}$. Proceeding similarly as in the proof of Proposition~\ref{P-supp} we have the following.
\begin{prop}
$\supp(\widehat{\bC}) \subset \widehat{S}^\sperp$, and thus $\supp(\bC) \subset S: = \{\bE\bE'\mid \bE \in \supp(\bG^{-1}), \, \bE' \in \widehat{S}^\sperp \}$.
\end{prop}

\begin{cor}\label{C-corSemi}
Let $\bC$ be a unitary matrix and $S$ be a MCS. If $\bC$ fixes $S$ pointwise then $\supp(\bC) \subset S$. The converse is also true. 
This property characterizes semi-Clifford matrices up to multiplication by Clifford. In particular, as for Hermitian Clifford matrices, we have that $\bC$ is supported on a \emph{commutative} subgroup.
\end{cor}

In the reminder of this section we show that Corollary~\ref{C-corSemi} holds for the entire third level $\cC^{(3)}$ of the hierarchy (always up to Cliffords). 
Note that this is not a trivial step because, as mentioned before, there exist elements in $\cC^{(3)}$ that are not semi-Cliffords~\cite{BS08}.

\begin{lem}
\label{lem:C3fixesPauli}
For $\bC \in \cC^{(3)}$ there exists a Pauli $\widetilde{\bE}$ such that $\bC\widetilde{\bE}\bC^{\dagger}$ is also a Pauli. As a consequence, there exists a Clifford correction $\bG$ such that $\bG\bC$ fixes (i.e., commutes with) some Pauli matrix.
\end{lem}

\begin{proof}
Let $\bC \in \cC^{(3)}$ and consider the map
\begin{equation*}\label{e-phiC}
    \varphi_\bC: \left\{ \begin{array}{ccccc} \cP\HW & \xrightarrow{\phi_\bC} & \cl_N & \xrightarrow{\,\,\Phi\,\,} & \Sp(2m) \\ \bE & \longmapsto & \bC\bE\bC^\dagger & \longmapsto & \Phi(\bC\bE\bC^\dagger) \end{array}\right.
\end{equation*}
where $\Phi$ is the map from~\eqref{e-Phi} and $\cP\HW$ is the projective form of $\HW$ that ignores phases. 
Then $\ker \varphi_\bC \subset \cP\HW$ has size $2^k$ for some $k \geq 0$. 
So, we see that $G := \im \varphi_\bC \subset \Sp(2m)$ has size $2^{2m-k}$. 

Let $G$ act on $\Fm \setminus \{{\bf 0}\}$. 
Since the size of an orbit must divide the size of the group $G$, each orbit has size a power of $2$ as well. 
However, since the orbits partition a set of size $2^{2m} - 1$ (odd), there must exist an orbit of odd size, and that size must be $2^0 = 1$. 
This means that there exists ${\bf 0} \neq \bc \in \Fm$ that is fixed by all symplectic matrices $\Phi(\bC\bE\bC^\dagger)$. The definition of $\Phi$ yields that the Hermitian Pauli $\widetilde{\bE} := \bE(\bc)$ either commutes or anticommutes with all $\bC\bE\bC^\dagger$.
In other words
\begin{equation}\label{e-wececd}
    \widetilde{\bE}\bC\bE\bC^\dagger = \alpha_\bE\bC\bE\bC^\dagger\widetilde{\bE}, \quad \alpha_\bE = \pm 1,
\end{equation}
for all Paulis $\bE$. Now put $\widetilde{\bC} = \widetilde{\bE}\bC\widetilde{\bE}$. Let also $\sigma_\bE = \pm 1$ be such that $\widetilde{\bE}\bE = \sigma_\bE \bE\widetilde{\bE}$. We have that
\begin{equation}\label{e-wececd1}
    \widetilde{\bE}\bC = \widetilde{\bC}\widetilde{\bE} \text{\quad and\quad} \widetilde{\bE}\bC^\dagger = \widetilde{\bC}^\dagger\widetilde{\bE},
\end{equation}
which combined with~\eqref{e-wececd} yields
\begin{equation}
    \alpha_\bE\bC\bE\bC^\dagger = \sigma_\bE\widetilde{\bC}\bE\widetilde{\bC}^\dagger
\end{equation}
for all Paulis $\bE$. 
Now it is easy to see that $\bC\widetilde{\bC}^\dagger$ is a Pauli $\bE'$, and thus~\eqref{e-wececd1} implies $\bC\widetilde{\bE}\bC^\dagger = \bE'\widetilde{\bE}$. 
\end{proof}

\begin{rem}
The proof of Lemma~\ref{lem:C3fixesPauli} could have been concluded using the language of stabilizer codes.
With the same notation, let $\widetilde{\bE}$ be such that it either commutes or anticommutes with all $\bC\bE\bC^\dagger \in \cl_N$. 
Now consider the stabilizer group $S = \l \widetilde{\bE}\r$ and the corresponding stabilizer codes $\cF_\varepsilon := \cF_\varepsilon(S)$, for $\varepsilon = \pm 1$ (see~\eqref{e-qecc}). 
We have that $\bC\bE\bC^\dagger \cF_\varepsilon = \pm \cF_\varepsilon$ for all Paulis $\bE$.
In other words, for any $v \in \cF_\varepsilon$, we have that  
\begin{equation}\label{e-CECd}
    \varepsilon \widetilde{\bE} \cdot \bC\bE\bC^\dagger v = \pm\left(\bC\bE\bC^\dagger \cdot \varepsilon \widetilde{\bE} v\right) = \pm \bC\bE\bC^\dagger v
\end{equation}
holds for all $\bE$. Next, express $\bC = \sum_\bE \alpha_\bE \bE$ and sum both sides in~\eqref{e-CECd} to obtain
\begin{equation}
    \varepsilon \widetilde{\bE}\left(\sum_\bE\alpha_\bE \bE\right)v = \varepsilon'\left(\sum_\bE\alpha_\bE \bE\right)v.
\end{equation}
In other words, $\varepsilon \widetilde{\bE}\bC v = \varepsilon'\bC v$, and thus $\bC$ permutes $\{\cF_\varepsilon\}$. 
This is equivalent with $\bC$ mapping $\widetilde{\bE}$ to some other Pauli.
\end{rem}

Next we prove the main result about the support of gates in $\cC^{(3)}$, which can then be straightforwardly used to show that every gate in $\cC^{(3)}$ is generalized semi-Clifford. Recall that a generalized semi-Clifford matrix is a unitary matrix that maps under conjugation the span of some MCS to the span of some other MCS.

\begin{theo}\label{T-main}
Let $\bC$ be a unitary matrix from $\cC^{(3)}$.
Then there exists a Clifford $\bG$ such that $\bG \bC$ is supported on a maximal commutative subgroup of $\HW$.
\end{theo}
\begin{proof}
From Lemma~\ref{lem:C3fixesPauli} we know that there exists some Clifford $\bH$ such that $\bH \bC$ commutes with some $\bE \in \HW$.
Now consider the group $S = \langle \bE \rangle$ and its normalizer in $\HW$ which we denote by $S^{\sperp}$.
Then, since $\bH \bC$ preserves $S$ under conjugation, it is a valid logical operator for the $[\![ m,m-1 ]\!]$ stabilizer code defined by $S$, i.e., it maps code states ($+1$ eigenvectors of $\bE$) to code states.
Denote this logical $(m-1)$-qubit operation realized by $\bH \bC$ as $\overline{\bC}_{\bH}$. Since $\bH\bC \in \cC^{(3)}$, we know that $\bH\bC$ must satisfy the necessary conjugation conditions on logical Paulis as determined by $\overline{\bC}_\bH$. These conditions can only correspond to physical realizations of logical Clifford gates (since physical Cliffords cannot realize anything above the second level $\cl_N = \cC^{(2)}$). Hence, we conclude that $\overline{\bC}_\bH \in \cC^{(3)}$ (in the logical space).

First, it has already been shown in~\cite{ZCC08} that operations in $\cC^{(3)}$ are semi-Clifford for $1$ and $2$ qubits.
Using Corollary~\ref{C-corSemi}, this automatically means that up to some Clifford correction such gates are supported on a MCS.
Therefore, we consider the induction hypothesis that for $(m-1)$ qubits, any $\cC^{(3)}$ element is supported on a MCS, up to multiplication by some Clifford.
Applying this hypothesis for $\overline{\bC}_{\bH}$ above, we see that there exists some $(m-1)$-qubit logical Clifford $\overline{\bG}'$ such that $\overline{\bG}'\, \overline{\bC}_{\bH}$ is supported on a MCS (of size $2^{m-1}$).
Note that a logical Clifford operation is defined by its action on logical Paulis.

Let this MCS be generated by logical $(m-1)$-qubit Paulis $\overline{\bE}_1, \overline{\bE}_2, \ldots, \overline{\bE}_{m-1}$, and let $\bE_1, \bE_2, \ldots, \bE_{m-1} \in S^{\sperp}$ form their respective physical $m$-qubit realizations in $\HW$.
These realizations are automatically defined once $2(m-1)$ $m$-qubit Pauli operations are chosen to be the appropriate physical realizations of the logical $X$ and $Z$ on the $(m-1)$ logical qubits.
As $\overline{\bG}'\, \overline{\bC}_{\bH}$ is supported on $\overline{\bE}_i$'s, it clearly commutes with each one of them.
Hence, by taking $\bG'$ to be an $m$-qubit Clifford that forms a physical realization of $\overline{\bG}'$, we see that $\bG' (\bH \bC)$ commutes with each $\bE_i$.
Note that, by definition of realizing a logical gate, such a $\bG'$ must preserve the stabilizer $S$, i.e., commute with $\bE$, and act on $\bE_i$ as $\overline{\bG}$ acts on $\overline{\bE}_i$.

Finally, consider the group $\langle \bE, \bE_1, \bE_2, \ldots, \bE_{m-1} \rangle$.
This is clearly a MCS and $(\bG' \bH) \bC$ fixes it pointwise. 
Therefore, by applying Corollary~\ref{C-corSemi} we see that $\bG \bC := (\bG' \bH) \bC$ is supported on this MCS.
This completes the induction. 
\end{proof}

Lemma~\ref{lem:C3fixesPauli} constitutes a crucial property of the third level $\cC^{(3)}$ of the Clifford hierarchy. 
This property is of course exclusive to the third level because we highly rely on the fact that $\im \phi_\bC$ is a subgroup of $\cl_N$. 
This in turn enables an induction argument on the number of qubits, rather than the typical induction arguments on the levels of the Clifford hierarchy. It is worth mentioning that the ``up to Clifford'' is indeed necessary throughout the paper. For instance, with regards to Lemma~\ref{lem:C3fixesPauli}, the physical Clifford permutation
\begin{equation}
    \bG = \left[\!\!\begin{array}{cccc}
         1&0&0&0 \\
         0&0&0&1 \\
         0&1&0&0 \\
         0&0&1&0
    \end{array}
    \!\!\right]
\end{equation}
does not fix (commute with) \emph{any} Pauli matrix. Similarly, one can easily produce other instances of examples that require the ``Clifford correction''.

If $\bC \in \cC^{(3)}$ is supported on a MCS $S$ then it trivially fixes the span of $S$, and therefore $\bC$ is a generalized semi-Clifford matrix.
The converse is not true since, for instance, even the $\rm{CNOT}$ gate fixes $Z_4$ but obviously is not supported on $Z_4$. 
In fact, any gate $\bG_D(\bP)$ fixes $Z_N$ by construction but is not supported on $Z_N$ (see Proposition~\ref{P-lem}(1)).

\begin{cor}[\mbox{\cite[Thm.~1.1]{BS08}}]
\label{C-BS08}
Every $\bC \in \cC^{(3)}$ is a generalized semi-Clifford matrix.
\end{cor}
\begin{proof}
By Theorem~\ref{T-main}, there exists a Clifford correction $\bG$ such that $\bG\bC$ is supported on a MCS $S$. 
Then $\bG\bC$ 
fixes the span of $S$, and $\bC$ maps the span of $S$ to the span of $\bG^\dagger S \bG$. Thus $\bC$ is generalized semi-Clifford. 
\end{proof}



\section{Conclusions and Future Research}

In this paper we study the Clifford hierarchy via the Pauli/Weyl expansion/support. 
First, we consider the Clifford group, that is, the second level of the hierarchy. We show that every element of the group is supported on a subgroup of the Pauli group (or a coset, when traceless).
Additionally, we give a closed form description of the support of standard group elements, and show that the coefficients are determined by a quadratic form modulo 4.
We argue that the Hermitian elements of the group play a prominent role on understanding the third level of the hierarchy. 
For this reason, we treat them separately, and among other things, we show that they are supported on a \emph{commutative} subgroup of the Pauli group (or a coset). 
Next, we consider the third level of the hierarchy. Our treatment is up to Clifford equivalence, which, at any rate, preservers the levels of the hierarchy. 
We show that, up to such equivalence, every third level matrix commutes with at least one Pauli matrix. 
This constitutes the main building block of a powerful induction argument on the number of qubits, which we use to prove that every third level matrix, up to Clifford equivalence, is supported on a maximal commutative subgroup of the Pauli group.
We believe that such induction argument can be further exploited in various aspects of quantum computation and quantum error-correction.

In future research, we will consider the behaviour of the support under elementary transformations such as multiplication and conjugation (which, surprisingly, is unknown). 
This, among other things, would give a closed form description of the support of \emph{any} Clifford group element. 
Next, with the ultimate goal of completely characterizing the third level of the hierarchy, we will consider the converse. 
That is, finding sufficient conditions under which a unitary $\bU$, supported on some MCS, belongs to the third level. 
We expect the coefficients to be (scaled) eighth roots of unity that are perhaps determined by third-order Reed-Muller codes. Finally, we will use the structural results of this paper to develop flag gadgets for third level operators, as well as reduce circuit complexity for these operators.


\section*{Acknowledgements}
The work of TP and OT was supported in part by the Academy of Finland under the grant 334539. The work of NR and RC was supported in part by NSF under the grant CCF1908730. TP and OT thank Robert Calderbank for his hospitality during visits to Duke University.

\bibliographystyle{plainnat}
\bibliography{IEEEabrv,ITW2020}

\end{document}